%% file: main-10.tex
\documentclass[conference]{IEEEtran}
\IEEEoverridecommandlockouts

\usepackage{cite}
\usepackage{algorithmic}
\usepackage{graphicx}
\usepackage{textcomp}
\usepackage{xcolor}

\newcommand{\bs}[1]{\boldsymbol{#1}}  
\newcommand{\mcl}[1]{\mathcal{#1}}

\makeatletter
\newcommand{\vast}{\bBigg@{3}} 
\newcommand{\Vast}{\bBigg@{4}} 
\makeatother

\newcommand{\hatP}{\widehat P}
\newcommand{\one}{\mathbf{1}}

\usepackage[utf8]{inputenc}
\usepackage[T1]{fontenc}
\usepackage{url}
\usepackage{tikz}
\usepackage{pgfplots} 
\usepgfplotslibrary{groupplots} 
\usepgfplotslibrary{dateplot}   
\pgfplotsset{compat=newest}
\usepackage{pgf}  
\usepackage{ifthen}  
\usepackage{cite}
\usepackage[cmex10]{amsmath} 
\usepackage{amsfonts, amssymb} 
\usepackage[version=4]{mhchem}
\usepackage{graphicx, amsthm, float, dsfont,color}
\usepackage[utf8]{inputenc} 
\usepackage[T1]{fontenc}
\usepackage{url}
\usepackage{ifthen}
\usepackage{cite}

\newtheorem{theorem}{Theorem} 
\newtheorem{lemma}{Lemma}
\newtheorem{corollary}{Corollary}  
 
\newtheorem{proposition}{Proposition}

 \DeclareMathOperator{\cav}{cav}

\begin{document}

\title{An Upper Bound to the Undetected Error Exponent with Erasure Decoding at Positive Rates
\thanks{This research was supported in part by the European Research Council under Grants 101142747 and 101158232, the Spanish Government under Grants PID2020-116683GB-C22, PID2021-128373OB-I00 and PID2025-167766OA-C22.}}

\author{\IEEEauthorblockN{Seyed AmirPouya Moeini}
\IEEEauthorblockA{
University of Cambridge\\
sam297@cam.ac.uk}
\and
\IEEEauthorblockN{Marco Dalai}
\IEEEauthorblockA{
University of Brescia\\
Brescia, Italy\\
marco.dalai@unibs.it}
\and
\IEEEauthorblockN{Albert Guill\'en i F\`abregas}
\IEEEauthorblockA{University of Cambridge\\
Universitat Polit\`ecnica de Catalunya\\
guillen@ieee.org}}

\maketitle

\begin{abstract}
	We study the undetected-error exponent of discrete memoryless channels with erasure decoding.
	Assuming that the undetected-error and erasure probabilities vanish, we derive an upper bound on the undetected-error exponent at positive rates. 
	The bound uses Somekh-Baruch's genie-aided argument where an auxiliary output provides the decoder with a list of codewords.
	It recovers Telatar's tight rate-zero converse and is at least as tight as the straight-line bound.
	A numerical example shows that the improvement over the straight-line bound can be strict.
\end{abstract}

\section{Introduction}

Let $W$ be a discrete memoryless channel (DMC) with finite input and output alphabets $\mcl X$ and $\mcl Y$, respectively.
A blocklength-$n$ code with $M_n$ messages consists of a codebook $\mcl C_n=\{\bs x_1,\ldots,\bs x_{M_n}\}\subseteq\mcl X^n$ and disjoint decoding regions $\mcl D_1,\ldots,\mcl D_{M_n}\subseteq\mcl Y^n$.
The decoder declares an erasure on the remaining outputs, which form the region
\begin{align}
 \mcl D_X=\mcl Y^n\setminus\bigcup_{m=1}^{M_n}\mcl D_m.
 \label{eq:erasure-region}
\end{align}
The maximal undetected-error and erasure probabilities are
\begin{align}
 P_e(\mcl C_n)&=\max_m\, W^n\!\bigg(\,\bigcup_{\bar m\ne m}\mcl D_{\bar m}\,\big|\,\bs x_m\bigg),
 \label{eq:undetected-probability}\\
 P_X(\mcl C_n)&=\max_m\, W^n(\mcl D_X|\bs x_m).
 \label{eq:erasure-probability}
\end{align}

We seek an upper bound on the undetected-error exponent for every sequence of codes with $P_e(\mcl C_n)\to0$ and $P_X(\mcl C_n)\to0$. 
At rate $R$, let $E_{\mathrm{u}}(R)$ denote the supremum of achievable undetected-error exponents over such sequences with $M_n\to\infty$ and $\liminf_{n\to\infty}\frac1n\log M_n\ge R$.

We assume $W(y|x)>0$ for all $(x,y)$.
This excludes channels with positive zero-undetected-error capacity, for which the undetected-error exponent is infinite at low rates.
Achievable error and erasure exponents were studied by Forney~\cite{Forney1968} and later by Telatar and Gallager~\cite{TelatarGallager1994} and Merhav~\cite{Merhav2008}. 
At rate zero, Telatar~\cite[Thm.~4.4]{Telatar1992} derived a tight converse on the optimal undetected error exponent, showing that
\begin{align}
\lim_{R \downarrow 0} E_{\mathrm{u}}(R)=\mcl F_W,
\label{eq:intro-zero-rate}
\end{align}
where $\mcl F_W\triangleq \max_Q\mcl F_W(Q)$, and
\begin{align}
 \mcl F_W(Q)
 \triangleq\sum_{(x,\bar x)}Q(x)Q(\bar x)\,
 D\bigl(W(\cdot|x)\|W(\cdot|\bar x)\bigr).
 \label{eq:F}
\end{align}
To the best of our knowledge, the only known converse at positive rates follows from the work of Nakibo\u glu and Zheng~\cite[Sec.~IV-C]{NakibogluZheng2012}.
They derived a straight-line bound for error-and-erasure decoding with feedback.
Applying their argument without feedback and using Telatar's rate-zero exponent as a starting point gives the following straight-line converse
\begin{align}
 E_{\mathrm{u}}(R)
 \le\left(1-\frac RC\right)  \mcl F_W,
 \quad 0\le R\le C,
 \label{eq:intro-straight-line}
\end{align}
where $C$ is the channel capacity.

In the standard no-erasure case, $\mcl D_X=\varnothing$ and $P_e(\mcl C_n)$ is the maximal decoding error probability. 
The reliability function $E(R,W)$ is the supremum of achievable error exponents at rate $R$ in this case. 
In~\cite{SomekhBaruch2024}, Somekh-Baruch derived an upper bound on $E(R,W)$ by introducing a genie that provides the decoder with a list containing the transmitted codeword, 
where all codewords in the list have the same joint type with an auxiliary sequence.
It was further shown in \cite{SomekhBaruch2024} that this bound is at least as tight as the classical converse bounds for DMCs.

In this work, we extend this approach to decoding with erasures and derive an upper bound on the undetected error exponent at positive rates.
We show that our bound is at least as tight as the straight-line bound \eqref{eq:intro-straight-line} and provide a numerical example showing a strict improvement.
Section~\ref{sec:main} presents the main results, Section~\ref{sec:proof} proves Theorem~\ref{thm:main}, and the Appendix contains the remaining proofs.

\smallskip\noindent\textit{Notation.} Scalar random variables are denoted by uppercase letters, their realizations by lowercase letters, and their alphabets by calligraphic letters. Random vectors are written in boldface. 
The type of a sequence $\bs x=(x_1,\ldots,x_n)\in\mcl X^n$ is its empirical distribution,
$ \hatP_{\bs x}(x)=\frac1n\sum_{i=1}^n\one\{x_i=x\}$.
Joint and conditional types are defined similarly. The type class $\mcl T^n(Q)$ contains all sequences in $\mcl X^n$ with type $Q$. 
When writing $\mcl C_n\subseteq\mcl T^n(Q)$ for a fixed $Q$, we allow the common $n$-type of the codewords to converge to $Q$.

\section{Main Results}\label{sec:main}

We extend Telatar's upper bound to positive rates using the genie construction of~\cite{SomekhBaruch2024}. 
Fix an input distribution $Q$ and a finite auxiliary alphabet $\mcl Z$. 
Let $P_{YZ|X}$ be an auxiliary broadcast channel with $P_{Y|X}=W$, and let $P_{XYZ}=QP_{YZ|X}$.
Write
\begin{align}\label{eq:conditional-channel}
 P_{Z|X}(z|x)
 &=\sum_y P_{YZ|X}(y,z|x),\\
 P_{Y|XZ}(y|x,z)
 &=\frac{P_{YZ|X}(y,z|x)}{P_{Z|X}(z|x)}. \label{eq--8}
\end{align}
Choose $\bs X$ equiprobably from $\mcl C_n$ and generate $(\bs Y,\bs Z)$ according to $P_{YZ|X}^n$.
Given $\bs Z=\bs z$ and the joint type $\hatP_{\bs{XZ}}$ of $(\bs X,\bs Z)$, the genie provides the decoder with the list
\begin{align}
 \mcl L=\{m:\hatP_{\bs x_m\bs z}=\hatP_{\bs{XZ}}\},\quad L=|\mcl L|.
 \label{eq:list}
\end{align}
The list contains all codewords having the same joint type with $\bs z$ as the transmitted codeword.
The list is random, since both $\bs Z$ and $\hatP_{\bs{XZ}}$ are random. 
Conditional on these observations, the transmitted message is equiprobable over $\mcl L$, and the output distribution for message $m\in\mcl L$ is $P_{Y|XZ}^n(\cdot|\bs x_m,\bs z)$.
The following lemma, which follows from~\cite[Lem.~6]{SomekhBaruch2024} and the law of large numbers, shows that the list is exponentially large with high probability whenever $I_P(X;Z)<R$.

\begin{lemma}
\label{lem:list}
	Let $\{\mcl C_n\}$ be a sequence of codes with $\mcl C_n\subseteq\mcl T^n(Q)$ and $\liminf_{n\to\infty}\frac1n\log M_n\ge R$.
	If $I_P(X;Z)<R$, then, for every $\delta>0$, with high probability,
\begin{align}
	 \hatP_{\bs{XZ}}&=P_{XZ}+o(1),\\
	 L&\ge\exp\bigl(n[R-I_P(X;Z)-\delta]\bigr).
	 \label{eq:large-list}
\end{align}
\end{lemma}
The genie reduces the decoding problem to the list $\mcl L$, which contains exponentially many messages by Lemma~\ref{lem:list}.
Since the additional information can only help, an upper bound on the optimal undetected error exponent for this genie-aided problem is also valid for the original code.


For any channel $V$ with input alphabet $\mcl X$, let
$\mcl F_V$ be defined as in \eqref{eq:F} with $W$ replaced by $V$.
We denote its upper concave envelope, as a function of the input
distribution, by $\operatorname{cav}\mcl F_V$.
For $R>0$, define
\begin{align}\label{eq:main-bound}
	\begin{split}
	 \!\!\overline{E}_{\mathrm{e}}(Q,R)
	 &\!\triangleq \!\!\!\!\! \inf_{\substack{\mcl Z,\,P_{YZ|X}\\ P_{Y|X}=W\\ I_P(X;Z)<R}} \!\!\!
	 \Bigg[\!\sum_z P_Z(z)\operatorname{cav} \mcl F_{P_{Y|XZ=z}}\!(P_{X|Z=z})\\
	 &\qquad\qquad\qquad+I_P(X;Z)-R\Bigg],
	\end{split}
\end{align}
where the infimum is over finite auxiliary alphabets and strictly positive auxiliary channels.

\begin{theorem}\label{thm:main}
Let $R>0$, and let $\mcl C_n\subseteq\mcl T^n(Q)$ be a sequence of codes
with $M_n$ codewords, satisfying
$P_e(\mcl C_n)\to0$, $P_X(\mcl C_n)\to0$, and
$\liminf_{n\to\infty}\frac1n\log M_n\ge R$.
Then,
\begin{align}
    \limsup_{n\to\infty}-\frac1n\log P_e(\mcl C_n)
    \le\overline E_{\mathrm e}(Q,R).
    \label{eq:cc-converse}
\end{align}
\end{theorem}
The result also extends to codes that are not necessarily of constant composition.
\begin{corollary}\label{cor:arbitrary}
For $0< R< C$,
\begin{align}
 E_{\mathrm{u}}(R)
 \le\sup_{Q:\,I(QW)\ge R}\overline{E}_{\mathrm{e}}(Q,R).
 \label{eq:arbitrary-converse}
\end{align}
\end{corollary}
We first show that our bound recovers Telatar's rate-zero converse.
For every $R>0$, a constant $Z$ satisfies the constraint $I_P(X;Z)<R$ in \eqref{eq:main-bound}, giving
\begin{align}
 \overline E_{\mathrm e}(Q,R)
 \le\operatorname{cav}\mcl F_W(Q)-R. 
 \label{eq:telatar-recovered}
\end{align}
Since $\max_Q\operatorname{cav}\mcl F_W(Q)=\mcl F_W$, Corollary~\ref{cor:arbitrary} gives $E_{\mathrm u}(R)\le\mcl F_W-R$, and letting $R\downarrow0$ recovers Telatar's bound \eqref{eq:intro-zero-rate}.

We next recover the straight-line converse bound in \eqref{eq:intro-straight-line} using an auxiliary channel adapted from~\cite[Eq.~(25)]{SomekhBaruch2024}.
\begin{proposition}\label{prop:line}
For $0<R\le I(QW)$,
\begin{align}
 \overline E_{\mathrm e}(Q,R)
 &\le\left(1-\frac RC\right) \mcl F_W.
 \label{eq:line-recovered}
\end{align}
\end{proposition}
Hence, Corollary~\ref{cor:arbitrary} recovers the straight-line bound in \eqref{eq:intro-straight-line}.
The proofs of Proposition~\ref{prop:line} and Corollary~\ref{cor:arbitrary} are given in the Appendix.

\subsection{Numerical Example}

We illustrate the improvement over the straight-line bound for a BSC with crossover probability $p=0.1$.

We consider a binary auxiliary alphabet $\mcl Z=\{0,1\}$ and choose an auxiliary channel of the symmetric form
\begin{align}
 \bigl(P_{Z|XY}(1|x,y)\bigr)_{x,y}
 =\begin{pmatrix}
 a_R&b_R\\
 1-b_R&1-a_R
 \end{pmatrix},
 \label{eq:optimized-kernel}
\end{align}
with $P_{YZ|X}=WP_{Z|XY}$, where the parameters $a_R,b_R\in(0,1)$ may depend on the rate $R$.
For binary inputs, $\mcl F_V(Q)$ is concave in $Q$, so that $\operatorname{cav}\mcl F_V=\mcl F_V$.
For each fixed auxiliary channel of this form, symmetry and concavity imply that the expression in~\eqref{eq:main-bound} is maximized by the uniform input distribution.
Thus, evaluating it at the uniform input gives an upper bound for arbitrary codes.

To choose the parameters, let $r\in(p,1/2)$ satisfy $R=\log2-h(r)$ and impose $(1-p)a_R+pb_R=r$.
This makes $P_{Z|X}$ a BSC with crossover probability $r$ and leaves one parameter to optimize numerically.
The boundary choice $I_P(X;Z)=R$ is approached through channels satisfying $I_P(X;Z)<R$.

Fig.~\ref{fig:bsc} compares the resulting upper bound with the straight-line converse and Telatar's achievable bound~\cite[Sec.~4.4.1]{Telatar1992}, which for this channel is given by
\begin{align}
 E_{\mathrm u}(R)
 \ge(r-p)\log\frac{1-p}{p},
 \label{eq:bsc-achievability}
\end{align}
where $r\in(p,1/2)$ is the unique solution of $R=\log2-h(r)$.
At $R=C/2$, the choice $a_R\simeq0.143679$ and $b_R\simeq0.770647$ gives an upper bound of approximately $0.3947$, which is strictly below the straight-line bound of $0.4394$.

\begin{figure}[t]
 \centering
 \input{bsc-figure.tex}
 \caption{Bounds on the undetected-error exponent for the BSC$(0.1)$: our upper bound (solid), the straight-line converse (dashed), and Telatar's achievable bound (dash-dotted).}
 \label{fig:bsc}
\end{figure}
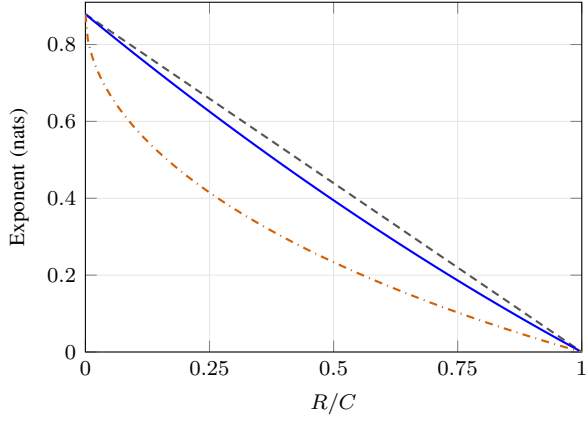

\section{Proof of Theorem~\ref{thm:main}}\label{sec:proof}

We first show that, with high probability, most messages in the genie list have small conditional probability of error or erasure.

\begin{lemma}\label{lem:reliable-sublist}
Under the assumptions of Theorem~\ref{thm:main}, 
for every $0<\delta<1$, with high probability, the list $\mcl L$ contains a sublist $\mcl L'$ such that
\begin{align}
 L'\triangleq|\mcl L'|&\ge(1-\delta)L,\\
 P_{Y|XZ}^n(\mcl D_m|\bs x_m,\bs z)&\ge1-\delta,
 \quad m\in\mcl L'.
 \label{eq:reliable-list}
\end{align}
\end{lemma}
\begin{proof}
Given the genie observation $(\bs Z,\hatP_{\bs{XZ}})=(\bs z,\hat P)$, let
\begin{align}
 G(\bs z,\hat P)
 =\frac1L\sum_{m\in\mcl L}
 \bigl[1-P_{Y|XZ}^n(\mcl D_m|\bs x_m,\bs z)\bigr].
 \label{eq:conditional-failure}
\end{align}
Since the transmitted message is equiprobable over $\mcl L$, this is the conditional probability of an undetected error or an erasure.
Averaging over the genie observations gives
\begin{align}
 \mathbb{E}[G(\bs Z,\hatP_{\bs{XZ}})]
 &=\frac1{M_n}\sum_{m=1}^{M_n}
 \bigl[1-W^n(\mcl D_m|\bs x_m)\bigr]\\
 &\le P_e(\mcl C_n)+P_X(\mcl C_n)=o(1).
 \label{eq:average-failure}
\end{align}
By Markov's inequality,
\begin{align}
 \mathbb{P}\bigl[G(\bs Z,\hatP_{\bs{XZ}})>\delta^2\bigr]
 \le\frac{\mathbb{E}[G(\bs Z,\hatP_{\bs{XZ}})]}{\delta^2}
 =o(1).
 \label{eq:markov-sublist}
\end{align}
Whenever $G(\bs z,\hat P)\le\delta^2$, fewer than $\delta L$ messages have conditional probability of error or erasure greater than $\delta$.
Removing these messages gives the required sublist.
\end{proof}

We now proceed to prove Theorem~\ref{thm:main}.
Let $R>0$, and fix a strictly positive auxiliary broadcast channel $P_{YZ|X}$ with $P_{Y|X}=W$ and $I_P(X;Z)<R$.
Choose $0<\delta<\min\{1,R-I_P(X;Z)\}$.
By Lemmas~\ref{lem:list} and~\ref{lem:reliable-sublist}, with high probability, the list satisfies~\eqref{eq:large-list} and contains a sublist satisfying~\eqref{eq:reliable-list}.
Fix a genie observation $(\bs Z,\hatP_{\bs{XZ}})=(\bs z,\hatP_{XZ})$ for which~\eqref{eq:large-list} and~\eqref{eq:reliable-list} hold.
Here, $\hatP_{XZ}$ is the joint type of the transmitted codeword and $\bs z$.
By the definition of the list, every codeword in $\mcl L$ has this same joint type with $\bs z$.
For $m\in\mcl L$, let
\begin{align}
 P_m=P_{Y|XZ}^n(\cdot|\bs x_m,\bs z)
 \label{eq:output-distribution}
\end{align}
denote the conditional output distribution of message $m$, so that~\eqref{eq:reliable-list} reads $P_m(\mcl D_m)\ge1-\delta$.
For $m,\bar m\in\mcl L'$, define
\begin{align}
 E_{m\bar m}
 \triangleq D(P_m\|P_{\bar m}).
 \label{eq:pair-divergence}
\end{align}
We first use these divergences to lower-bound the probability of decoding one message as another.
For $m\ne\bar m$, let $u=P_m(\mcl D_m)$ and $v=P_{\bar m}(\mcl D_m)$.
By data processing \cite{Telatar1992} (see also \cite[Ch.~7]{PolyanskiyWu2025}),
\begin{align}
 E_{m\bar m}
 &\ge D(\operatorname{Bern}(u)\|\operatorname{Bern}(v))\\
 &\ge-h(u)-u\log v.
 \label{eq:binary-dp-step}
\end{align}
By~\eqref{eq:reliable-list}, we have $u\ge1-\delta$.
Together with $h(u)\le\log2$, 
\begin{align}
 P_{\bar m}(\mcl D_m)
 \ge\exp\left(-\frac{E_{m\bar m}+\log2}{1-\delta}\right).
 \label{eq:pair-error}
\end{align}
We next bound the average of these divergences over $\mcl L'$.
For each coordinate $i$, write $V_i(\cdot|x)=P_{Y|XZ}(\cdot|x,z_i)$.
By memorylessness, the divergence in~\eqref{eq:pair-divergence} is the sum of the coordinatewise divergences, that is,
\begin{align}
 E_{m\bar m}
 =\sum_{i=1}^n
 D\bigl(V_i(\cdot|x_{m,i})\|V_i(\cdot|x_{\bar m,i})\bigr).
 \label{eq:divergence-coordinates}
\end{align}
Let $Q_i$ be the distribution of the $i$th symbol when a message is chosen equiprobably from $\mcl L'$, that is,
\begin{align}
 Q_i(x)=\frac1{L'}|\{m\in\mcl L':x_{m,i}=x\}|.
 \label{eq:coordinate-distribution}
\end{align}
At coordinate $i$, exactly $L'Q_i(x)$ messages have symbol $x$.
Hence, the average of the $i$th term over all ordered pairs is
\begin{align}
\begin{split}
 &\frac1{L'^2}\sum_{m,\bar m\in\mcl L'}
 D\bigl(V_i(\cdot|x_{m,i})\|V_i(\cdot|x_{\bar m,i})\bigr)\\
 &\quad\qquad=\sum_{(x,\bar x)}Q_i(x)Q_i(\bar x)\,
 D\bigl(V_i(\cdot|x)\|V_i(\cdot|\bar x)\bigr)
\end{split}\\
 &\quad\qquad=\mcl F_{V_i}(Q_i).
 \label{eq:all-pair-coordinate}
\end{align}
Summing over the coordinates therefore gives
\begin{align}
 \frac1{nL'^2}\sum_{m,\bar m\in\mcl L'}E_{m\bar m}
 =\frac1n\sum_{i=1}^n\mcl F_{V_i}(Q_i).
 \label{eq:pair-identity}
\end{align}
We now group the coordinates according to $z_i$.
All codewords in $\mcl L'$ have joint type $\hatP_{XZ}$ with $\bs z$.
Thus, for each $(x,z)$,
\begin{align}
 \sum_{i:z_i=z}Q_i(x)
 &=\frac1{L'}\sum_{m\in\mcl L'}
 |\{i:(x_{m,i},z_i)=(x,z)\}|\\
 &=n\hatP_{XZ}(x,z).
 \label{eq:count-coordinate}
\end{align}
Dividing by $n\hatP_Z(z)$ gives, for $\hatP_Z(z)>0$,
\begin{align}
 \frac1{n\hatP_Z(z)}\sum_{i:z_i=z}Q_i
 =\hatP_{X|Z=z}.
 \label{eq:coordinate-average}
\end{align}
The conditional channel is the same at all coordinates with $z_i=z$, namely $V_i=P_{Y|XZ=z}$.
Bounding $\mcl F$ by its upper concave envelope and applying Jensen's inequality gives
\begin{align}
\begin{split}
	 &\frac{1}{n\hatP_Z(z)}
 \sum_{i:z_i=z}\mcl F_{P_{Y|XZ=z}}(Q_i)
 \\
 &\quad\qquad\le
 \operatorname{cav}\mcl F_{P_{Y|XZ=z}}
 \left(\frac{1}{n\hatP_Z(z)}
 \sum_{i:z_i=z}Q_i\right)
 \end{split}
 \\
 &\quad\qquad=
 \operatorname{cav}\mcl F_{P_{Y|XZ=z}}
 (\hatP_{X|Z=z}),
 \label{eq:group-concavity}
\end{align}
where the equality follows from~\eqref{eq:coordinate-average}.
Multiplying by $\hatP_Z(z)$ and summing over $z$ gives
\begin{align}
 \frac1n\sum_{i=1}^n\mcl F_{V_i}(Q_i)
 \le\sum_z\hatP_Z(z)\operatorname{cav}
 \mcl F_{P_{Y|XZ=z}}(\hatP_{X|Z=z}).
 \label{eq:concave-average}
\end{align}
Since $E_{mm}=0$, restricting the average in~\eqref{eq:pair-identity} to distinct messages introduces a factor $L'/(L'-1)$, which tends to one.
Also, $\hatP_{XZ}=P_{XZ}+o(1)$ by Lemma~\ref{lem:list}.
The auxiliary channel is fixed and strictly positive, so continuity gives
\begin{align}
\begin{split}
	 &\frac1{nL'(L'-1)}
 \sum_{\substack{m,\bar m\in\mcl L'\\m\ne\bar m}}E_{m\bar m}
 \\
 &\qquad\le\sum_zP_Z(z)\operatorname{cav}
 \mcl F_{P_{Y|XZ=z}}(P_{X|Z=z})+o(1),
 \label{eq:pair-average}
\end{split}
\end{align}
where the $o(1)$ term is uniform over the observations under consideration.
 
We now combine~\eqref{eq:pair-error} and~\eqref{eq:pair-average}.
The decoding regions are disjoint, and the transmitted message is equiprobable over the full list $\mcl L$.
Hence, conditioned on the genie observation, the undetected-error probability is at least
\begin{align}
 \frac1L\sum_{\bar m\in\mcl L'}
 \sum_{\substack{m\in\mcl L'\\m\ne\bar m}}
 P_{\bar m}(\mcl D_m).
 \label{eq:sum-regions}
\end{align}
Substituting~\eqref{eq:pair-error} and applying Jensen's inequality bounds this expression from below by
\begin{align}
 \frac{L'(L'-1)}L
 \exp\left(-\frac{1}{1-\delta}
 \left[\frac{\sum_{m\ne\bar m}E_{m\bar m}}{L'(L'-1)}+\log2\right]\right),
 \label{eq:jensen}
\end{align}
where the sum is over distinct messages in $\mcl L'$.
Since $L'\ge(1-\delta)L$, the prefactor has the same exponential rate as $L$, which is at least $R-I_P(X;Z)-\delta$ by~\eqref{eq:large-list}.
Averaging over the genie observations and using~\eqref{eq:pair-average}, we obtain
\begin{align} \label{eq:fixed-genie-bound}
\begin{split}
	 &\limsup_{n\to\infty}-\frac1n\log P_e(\mcl C_n)
 \\
 &\qquad\le\frac1{1-\delta}\sum_zP_Z(z)
 \operatorname{cav}\mcl F_{P_{Y|XZ=z}}(P_{X|Z=z})
\end{split}
 \\
 &\qquad\qquad\qquad+I_P(X;Z)-R+\delta. \nonumber
\end{align}
Here, the observations used above have probability $1-o(1)$, and the average undetected-error probability is at most $P_e(\mcl C_n)$.
Letting $\delta\downarrow0$ and taking the infimum over all feasible auxiliary channels and finite alphabets $\mcl Z$ gives~\eqref{eq:cc-converse}.

\appendix

\begin{proof}[Proof of Proposition~\ref{prop:line}]
Let $\mcl Z=\mcl Y\cup\{e\}$, with $e\notin\mcl Y$, and choose
$P_{YZ|X}(y,z|x)=W(y|x)P_{Z|XY}(z|x,y)$, where
\begin{align}
 P_{Z|XY}(z|x,y)
 =\begin{cases}
 1-\alpha,&z=e,\\
 \alpha[(1-\varepsilon)\delta_y(z)
 +\varepsilon W(z|x)],&z\in\mcl Y.
 \end{cases}
 \label{eq:reveal-kernel}
\end{align}
Here, $0<\alpha<R/I(QW)$, $0<\varepsilon<1$, and $\delta_y$ is the point mass at $y$.
For $\varepsilon=0$, this is the auxiliary channel in~\cite[Eq.~(25)]{SomekhBaruch2024}, and $\varepsilon>0$ makes it strictly positive.
Summing over $z$ gives $P_{Y|X}=W$, as required.
The other marginal is
\begin{align} \label{eq--17}
 P_{Z|X}(z|x)
 =\begin{cases}
 1-\alpha,&z=e,\\
 \alpha W(z|x),&z\in\mcl Y.
 \end{cases}
\end{align} 
Thus, $Z=e$ with probability $1-\alpha$, independently of $X$, and otherwise $Z$ is distributed according to $W(\cdot|X)$.
It follows that $I_P(X;Z)=\alpha I(QW)<R$.

When $Z=e$, we have $P_{X|Z=e}=Q$ and $P_{Y|XZ=e}=W$. For $z\in\mathcal{Y}$, the conditional channel follows from \eqref{eq--8}, \eqref{eq:reveal-kernel} and \eqref{eq--17} as
\begin{align}
  P_{Y|XZ}(y|x,z)
  &=\frac{W(y|x)\,\alpha\bigl[(1-\varepsilon)\delta_y(z)+\varepsilon W(z|x)\bigr]}{\alpha W(z|x)} \\
  &=(1-\varepsilon)\delta_z(y)+\varepsilon W(y|x),
  \label{eq:cond}
\end{align}
where we used that $\delta_y(z)W(y|x)=\delta_z(y)W(z|x)$.
This channel is a mixture of the point mass $\delta_z$, which does not depend on $x$, and $W$. By joint convexity of divergence, for every pair $(x,\bar x)$,
\begin{align}
  &D\bigl(P_{Y|XZ}(\cdot|x,z)\,\big\|\,P_{Y|XZ}(\cdot|\bar x,z)\bigr) \\
  &\quad\qquad\le(1-\varepsilon)D(\delta_z\|\delta_z)+\varepsilon D\bigl(W(\cdot|x)\,\big\|\,W(\cdot|\bar x)\bigr) \\
  &\quad\qquad=\varepsilon D\bigl(W(\cdot|x)\,\big\|\,W(\cdot|\bar x)\bigr).  \label{eq:conv}
\end{align}
Multiplying by $Q'(x)Q'(\bar x)$ and summing over $(x,\bar x)$ gives, for every input distribution $Q'$,
\begin{equation}
  \mathcal{F}_{P_{Y|XZ=z}}(Q')\le\varepsilon\mathcal{F}_W(Q')\le\varepsilon\cav\mathcal{F}_W(Q').
  \label{eq:Fbound}
\end{equation}
Since $\varepsilon\cav\mathcal{F}_W$ is concave, taking the upper concave envelope of the left-hand side gives
\begin{equation}
  \cav\mathcal{F}_{P_{Y|XZ=z}}(Q')\le\varepsilon\cav\mathcal{F}_W(Q'),\quad z\in\mathcal{Y}.
  \label{eq:cavbound}
\end{equation}
We now average over $Z$. By \eqref{eq--17}, $P_Z(e)=1-\alpha$, so the term $z=e$ contributes $(1-\alpha)\cav\mathcal{F}_W(Q)$ to the first term in \eqref{eq:main-bound}. 
Also, $P_{XZ}(x,z)=\alpha Q(x)W(z|x)$ for $z\in\mathcal{Y}$, so that
\begin{equation}
  \sum_{z\in\mathcal{Y}}P_Z(z)=\alpha,\quad
  \sum_{z\in\mathcal{Y}}P_Z(z)P_{X|Z=z}=\alpha Q.
  \label{eq:mean}
\end{equation}
The remaining terms therefore satisfy
\begin{align}
\begin{split} \label{eq:useC}
	  &\sum_{z\in\mathcal{Y}}P_Z(z)\cav\mathcal{F}_{P_{Y|XZ=z}}\bigl(P_{X|Z=z}\bigr) \\
  &\quad\qquad\le\varepsilon\sum_{z\in\mathcal{Y}}P_Z(z)\cav\mathcal{F}_W\bigl(P_{X|Z=z}\bigr)
\end{split}\\
  &\quad\qquad\le\varepsilon\alpha\cav\mathcal{F}_W\biggl(\frac{1}{\alpha}\sum_{z\in\mathcal{Y}}P_Z(z)P_{X|Z=z}\biggr) \label{eq:useJ}\\
  &\quad\qquad=\alpha\varepsilon\cav\mathcal{F}_W(Q),
  \label{eq:rest}
\end{align}
where \eqref{eq:useC} follows from \eqref{eq:cavbound}, \eqref{eq:useJ} from the concavity of $\cav\mathcal{F}_W$ and Jensen's inequality, and \eqref{eq:rest} from \eqref{eq:mean}. 
The first term in \eqref{eq:main-bound} is therefore at most $(1-\alpha+\alpha\varepsilon)\cav\mathcal{F}_W(Q)$. 
Substituting this and $I_P(X;Z)=\alpha I(QW)$ into \eqref{eq:main-bound} gives
\begin{equation}
  \overline{E}_{\mathrm e}(Q,R)\le(1-\alpha+\alpha\varepsilon)\cav\mathcal{F}_W(Q)+\alpha I(QW)-R.
\end{equation}
Letting $\varepsilon\downarrow0$ and $\alpha\uparrow R/I(QW)$, and using $I(QW)\le C$ and
$\operatorname{cav}\mcl F_W(Q)\le \mcl F_W$, yields~\eqref{eq:line-recovered}.
\end{proof}

\begin{proof}[Proof of Corollary~\ref{cor:arbitrary}]
Some type class contains at least $M_n/(n+1)^{|\mcl X|}$ codewords of $\mcl C_n$ (see~\cite[p.~42]{Telatar1992}).
Retain these codewords and their decoding regions, and turn the other decoding regions into erasures.
The rate loss is negligible, the undetected-error probability does not increase, and the erasure probability increases by at most $P_e(\mcl C_n)$, so it remains $o(1)$.
By passing to a subsequence, we may assume that the type of the retained codewords converges to some $Q$.
Since the retained codes have rate at least $R$ and are decoded reliably, Fano's inequality gives $I(QW)\ge R$.
By Theorem~\ref{thm:main}, their undetected-error exponent is at most $\overline E_{\mathrm e}(Q,R)$.
The same bound holds for $\mcl C_n$, since its undetected-error probability is at least as large.
\end{proof}

\section*{Acknowledgements}
The initial bound construction and proofs were developed by the authors. 
OpenAI GPT-6 Astra identified and helped correct a technical error in Lemma 2 and its use in Theorem 1. 
The model was also used to check the remaining proofs, refine the presentation of technical arguments, comparison with the straight-line bound, and assist with the numerical results. 
All AI-assisted outputs were independently verified by the authors.

{\footnotesize
\bibliographystyle{IEEEtran}
\bibliography{refs}
}

\end{document}

%% file: bsc-figure.tex
\definecolor{boundblue}{HTML}{0072B2}
\definecolor{boundgrey}{HTML}{555555}
\definecolor{boundorange}{HTML}{D55E00}

\begin{tikzpicture}
\begin{axis}[
 width=0.92\columnwidth,height=0.7\columnwidth,
 xlabel={$R/C$},ylabel={Exponent (nats)},
 xmin=0,xmax=1,ymin=0,ymax=0.91,
 xtick={0,0.25,0.5,0.75,1},ytick={0,0.2,0.4,0.6,0.8},
 tick label style={font=\footnotesize},label style={font=\footnotesize},
 grid=major,major grid style={draw=gray!20},
 legend style={font=\footnotesize,draw=none,fill=none,at={(0.98,0.98)},anchor=north east},
 legend cell align=left,
]
\addplot[boundgrey,densely dashed,thick] coordinates {
(0.00000,0.8788898309)
(0.00500,0.8744953818)
(0.01000,0.8701009326)
(0.01500,0.8657064835)
(0.02000,0.8613120343)
(0.02500,0.8569175852)
(0.03000,0.8525231360)
(0.03500,0.8481286869)
(0.04000,0.8437342377)
(0.04500,0.8393397885)
(0.05000,0.8349453394)
(0.05500,0.8305508902)
(0.06000,0.8261564411)
(0.06500,0.8217619919)
(0.07000,0.8173675428)
(0.07500,0.8129730936)
(0.08000,0.8085786445)
(0.08500,0.8041841953)
(0.09000,0.7997897462)
(0.09500,0.7953952970)
(0.10000,0.7910008478)
(0.10500,0.7866063987)
(0.11000,0.7822119495)
(0.11500,0.7778175004)
(0.12000,0.7734230512)
(0.12500,0.7690286021)
(0.13000,0.7646341529)
(0.13500,0.7602397038)
(0.14000,0.7558452546)
(0.14500,0.7514508054)
(0.15000,0.7470563563)
(0.15500,0.7426619071)
(0.16000,0.7382674580)
(0.16500,0.7338730088)
(0.17000,0.7294785597)
(0.17500,0.7250841105)
(0.18000,0.7206896614)
(0.18500,0.7162952122)
(0.19000,0.7119007631)
(0.19500,0.7075063139)
(0.20000,0.7031118647)
(0.20500,0.6987174156)
(0.21000,0.6943229664)
(0.21500,0.6899285173)
(0.22000,0.6855340681)
(0.22500,0.6811396190)
(0.23000,0.6767451698)
(0.23500,0.6723507207)
(0.24000,0.6679562715)
(0.24500,0.6635618224)
(0.25000,0.6591673732)
(0.25500,0.6547729240)
(0.26000,0.6503784749)
(0.26500,0.6459840257)
(0.27000,0.6415895766)
(0.27500,0.6371951274)
(0.28000,0.6328006783)
(0.28500,0.6284062291)
(0.29000,0.6240117800)
(0.29500,0.6196173308)
(0.30000,0.6152228817)
(0.30500,0.6108284325)
(0.31000,0.6064339833)
(0.31500,0.6020395342)
(0.32000,0.5976450850)
(0.32500,0.5932506359)
(0.33000,0.5888561867)
(0.33500,0.5844617376)
(0.34000,0.5800672884)
(0.34500,0.5756728393)
(0.35000,0.5712783901)
(0.35500,0.5668839410)
(0.36000,0.5624894918)
(0.36500,0.5580950426)
(0.37000,0.5537005935)
(0.37500,0.5493061443)
(0.38000,0.5449116952)
(0.38500,0.5405172460)
(0.39000,0.5361227969)
(0.39500,0.5317283477)
(0.40000,0.5273338986)
(0.40500,0.5229394494)
(0.41000,0.5185450003)
(0.41500,0.5141505511)
(0.42000,0.5097561019)
(0.42500,0.5053616528)
(0.43000,0.5009672036)
(0.43500,0.4965727545)
(0.44000,0.4921783053)
(0.44500,0.4877838562)
(0.45000,0.4833894070)
(0.45500,0.4789949579)
(0.46000,0.4746005087)
(0.46500,0.4702060595)
(0.47000,0.4658116104)
(0.47500,0.4614171612)
(0.48000,0.4570227121)
(0.48500,0.4526282629)
(0.49000,0.4482338138)
(0.49500,0.4438393646)
(0.50000,0.4394449155)
(0.50500,0.4350504663)
(0.51000,0.4306560172)
(0.51500,0.4262615680)
(0.52000,0.4218671188)
(0.52500,0.4174726697)
(0.53000,0.4130782205)
(0.53500,0.4086837714)
(0.54000,0.4042893222)
(0.54500,0.3998948731)
(0.55000,0.3955004239)
(0.55500,0.3911059748)
(0.56000,0.3867115256)
(0.56500,0.3823170765)
(0.57000,0.3779226273)
(0.57500,0.3735281781)
(0.58000,0.3691337290)
(0.58500,0.3647392798)
(0.59000,0.3603448307)
(0.59500,0.3559503815)
(0.60000,0.3515559324)
(0.60500,0.3471614832)
(0.61000,0.3427670341)
(0.61500,0.3383725849)
(0.62000,0.3339781358)
(0.62500,0.3295836866)
(0.63000,0.3251892374)
(0.63500,0.3207947883)
(0.64000,0.3164003391)
(0.64500,0.3120058900)
(0.65000,0.3076114408)
(0.65500,0.3032169917)
(0.66000,0.2988225425)
(0.66500,0.2944280934)
(0.67000,0.2900336442)
(0.67500,0.2856391951)
(0.68000,0.2812447459)
(0.68500,0.2768502967)
(0.69000,0.2724558476)
(0.69500,0.2680613984)
(0.70000,0.2636669493)
(0.70500,0.2592725001)
(0.71000,0.2548780510)
(0.71500,0.2504836018)
(0.72000,0.2460891527)
(0.72500,0.2416947035)
(0.73000,0.2373002544)
(0.73500,0.2329058052)
(0.74000,0.2285113560)
(0.74500,0.2241169069)
(0.75000,0.2197224577)
(0.75500,0.2153280086)
(0.76000,0.2109335594)
(0.76500,0.2065391103)
(0.77000,0.2021446611)
(0.77500,0.1977502120)
(0.78000,0.1933557628)
(0.78500,0.1889613137)
(0.79000,0.1845668645)
(0.79500,0.1801724153)
(0.80000,0.1757779662)
(0.80500,0.1713835170)
(0.81000,0.1669890679)
(0.81500,0.1625946187)
(0.82000,0.1582001696)
(0.82500,0.1538057204)
(0.83000,0.1494112713)
(0.83500,0.1450168221)
(0.84000,0.1406223729)
(0.84500,0.1362279238)
(0.85000,0.1318334746)
(0.85500,0.1274390255)
(0.86000,0.1230445763)
(0.86500,0.1186501272)
(0.87000,0.1142556780)
(0.87500,0.1098612289)
(0.88000,0.1054667797)
(0.88500,0.1010723306)
(0.89000,0.0966778814)
(0.89500,0.0922834322)
(0.90000,0.0878889831)
(0.90500,0.0834945339)
(0.91000,0.0791000848)
(0.91500,0.0747056356)
(0.92000,0.0703111865)
(0.92500,0.0659167373)
(0.93000,0.0615222882)
(0.93500,0.0571278390)
(0.94000,0.0527333899)
(0.94500,0.0483389407)
(0.95000,0.0439444915)
(0.95500,0.0395500424)
(0.96000,0.0351555932)
(0.96500,0.0307611441)
(0.97000,0.0263666949)
(0.97500,0.0219722458)
(0.98000,0.0175777966)
(0.98500,0.0131833475)
(0.99000,0.0087888983)
(0.99500,0.0043944492)
(1.00000,0.0000000000)
};
\addplot[blue,solid,thick] coordinates {
(0.00000,0.8788898309)
(0.00500,0.8736045717)
(0.01000,0.8683282796)
(0.01500,0.8630609547)
(0.02000,0.8578025968)
(0.02500,0.8525532059)
(0.03000,0.8473127819)
(0.03500,0.8420813247)
(0.04000,0.8368588341)
(0.04500,0.8316453101)
(0.05000,0.8264407526)
(0.05500,0.8212451613)
(0.06000,0.8160585361)
(0.06500,0.8108808770)
(0.07000,0.8057121838)
(0.07500,0.8005524562)
(0.08000,0.7954016942)
(0.08500,0.7902598975)
(0.09000,0.7851270660)
(0.09500,0.7800031995)
(0.10000,0.7748882978)
(0.10500,0.7697823607)
(0.11000,0.7646853880)
(0.11500,0.7595973795)
(0.12000,0.7545183350)
(0.12500,0.7494482541)
(0.13000,0.7443871368)
(0.13500,0.7393349827)
(0.14000,0.7342917917)
(0.14500,0.7292575634)
(0.15000,0.7242322976)
(0.15500,0.7192159940)
(0.16000,0.7142086524)
(0.16500,0.7092102725)
(0.17000,0.7042208540)
(0.17500,0.6992403966)
(0.18000,0.6942689000)
(0.18500,0.6893063639)
(0.19000,0.6843527881)
(0.19500,0.6794081721)
(0.20000,0.6744725158)
(0.20500,0.6695458186)
(0.21000,0.6646280804)
(0.21500,0.6597193008)
(0.22000,0.6548194795)
(0.22500,0.6499286161)
(0.23000,0.6450467102)
(0.23500,0.6401737616)
(0.24000,0.6353097697)
(0.24500,0.6304547344)
(0.25000,0.6256086552)
(0.25500,0.6207715318)
(0.26000,0.6159433637)
(0.26500,0.6111241506)
(0.27000,0.6063138921)
(0.27500,0.6015125879)
(0.28000,0.5967202375)
(0.28500,0.5919368405)
(0.29000,0.5871623965)
(0.29500,0.5823969052)
(0.30000,0.5776403662)
(0.30500,0.5728927789)
(0.31000,0.5681541431)
(0.31500,0.5634244583)
(0.32000,0.5587037241)
(0.32500,0.5539919401)
(0.33000,0.5492891059)
(0.33500,0.5445952210)
(0.34000,0.5399102850)
(0.34500,0.5352342975)
(0.35000,0.5305672581)
(0.35500,0.5259091664)
(0.36000,0.5212600219)
(0.36500,0.5166198242)
(0.37000,0.5119885729)
(0.37500,0.5073662676)
(0.38000,0.5027529079)
(0.38500,0.4981484932)
(0.39000,0.4935530233)
(0.39500,0.4889664976)
(0.40000,0.4843889158)
(0.40500,0.4798202774)
(0.41000,0.4752605821)
(0.41500,0.4707098294)
(0.42000,0.4661680189)
(0.42500,0.4616351502)
(0.43000,0.4571112229)
(0.43500,0.4525962367)
(0.44000,0.4480901910)
(0.44500,0.4435930855)
(0.45000,0.4391049199)
(0.45500,0.4346256937)
(0.46000,0.4301554066)
(0.46500,0.4256940582)
(0.47000,0.4212416481)
(0.47500,0.4167981760)
(0.48000,0.4123636415)
(0.48500,0.4079380443)
(0.49000,0.4035213841)
(0.49500,0.3991136604)
(0.50000,0.3947148731)
(0.50500,0.3903250217)
(0.51000,0.3859441061)
(0.51500,0.3815721258)
(0.52000,0.3772090806)
(0.52500,0.3728549703)
(0.53000,0.3685097945)
(0.53500,0.3641735531)
(0.54000,0.3598462457)
(0.54500,0.3555278722)
(0.55000,0.3512184324)
(0.55500,0.3469179260)
(0.56000,0.3426263529)
(0.56500,0.3383437129)
(0.57000,0.3340700059)
(0.57500,0.3298052316)
(0.58000,0.3255493901)
(0.58500,0.3213024811)
(0.59000,0.3170645046)
(0.59500,0.3128354605)
(0.60000,0.3086153488)
(0.60500,0.3044041694)
(0.61000,0.3002019223)
(0.61500,0.2960086076)
(0.62000,0.2918242252)
(0.62500,0.2876487752)
(0.63000,0.2834822577)
(0.63500,0.2793246727)
(0.64000,0.2751760205)
(0.64500,0.2710363010)
(0.65000,0.2669055146)
(0.65500,0.2627836613)
(0.66000,0.2586707414)
(0.66500,0.2545667552)
(0.67000,0.2504717029)
(0.67500,0.2463855848)
(0.68000,0.2423084013)
(0.68500,0.2382401527)
(0.69000,0.2341808394)
(0.69500,0.2301304618)
(0.70000,0.2260890205)
(0.70500,0.2220565158)
(0.71000,0.2180329483)
(0.71500,0.2140183186)
(0.72000,0.2100126273)
(0.72500,0.2060158749)
(0.73000,0.2020280622)
(0.73500,0.1980491899)
(0.74000,0.1940792588)
(0.74500,0.1901182695)
(0.75000,0.1861662230)
(0.75500,0.1822231201)
(0.76000,0.1782889618)
(0.76500,0.1743637489)
(0.77000,0.1704474826)
(0.77500,0.1665401639)
(0.78000,0.1626417937)
(0.78500,0.1587523735)
(0.79000,0.1548719042)
(0.79500,0.1510003872)
(0.80000,0.1471378238)
(0.80500,0.1432842153)
(0.81000,0.1394395631)
(0.81500,0.1356038688)
(0.82000,0.1317771338)
(0.82500,0.1279593597)
(0.83000,0.1241505482)
(0.83500,0.1203507009)
(0.84000,0.1165598196)
(0.84500,0.1127779062)
(0.85000,0.1090049626)
(0.85500,0.1052409906)
(0.86000,0.1014859923)
(0.86500,0.0977399699)
(0.87000,0.0940029254)
(0.87500,0.0902748611)
(0.88000,0.0865557794)
(0.88500,0.0828456825)
(0.89000,0.0791445730)
(0.89500,0.0754524534)
(0.90000,0.0717693263)
(0.90500,0.0680951943)
(0.91000,0.0644300604)
(0.91500,0.0607739272)
(0.92000,0.0571267978)
(0.92500,0.0534886752)
(0.93000,0.0498595626)
(0.93500,0.0462394630)
(0.94000,0.0426283799)
(0.94500,0.0390263166)
(0.95000,0.0354332766)
(0.95500,0.0318492636)
(0.96000,0.0282742811)
(0.96500,0.0247083330)
(0.97000,0.0211514232)
(0.97500,0.0176035557)
(0.98000,0.0140647346)
(0.98500,0.0105349641)
(0.99000,0.0070142485)
(0.99500,0.0035025923)
(1.00000,0.0000000000)
};
\addplot[boundorange,dashdotted,thick] coordinates {
(0.00000,0.8788898309)
(0.00500,0.8122593474)
(0.01000,0.7846890431)
(0.01500,0.7635533787)
(0.02000,0.7457517685)
(0.02500,0.7300828294)
(0.03000,0.7159302249)
(0.03500,0.7029277097)
(0.04000,0.6908365825)
(0.04500,0.6794909908)
(0.05000,0.6687701421)
(0.05500,0.6585828212)
(0.06000,0.6488581418)
(0.06500,0.6395397155)
(0.07000,0.6305818092)
(0.07500,0.6219467218)
(0.08000,0.6136029362)
(0.08500,0.6055237852)
(0.09000,0.5976864662)
(0.09500,0.5900712992)
(0.10000,0.5826611583)
(0.10500,0.5754410305)
(0.11000,0.5683976669)
(0.11500,0.5615193056)
(0.12000,0.5547954475)
(0.12500,0.5482166742)
(0.13000,0.5417744985)
(0.13500,0.5354612398)
(0.14000,0.5292699205)
(0.14500,0.5231941797)
(0.15000,0.5172281985)
(0.15500,0.5113666381)
(0.16000,0.5056045854)
(0.16500,0.4999375075)
(0.17000,0.4943612113)
(0.17500,0.4888718089)
(0.18000,0.4834656874)
(0.18500,0.4781394822)
(0.19000,0.4728900540)
(0.19500,0.4677144680)
(0.20000,0.4626099757)
(0.20500,0.4575739988)
(0.21000,0.4526041147)
(0.21500,0.4476980436)
(0.22000,0.4428536373)
(0.22500,0.4380688683)
(0.23000,0.4333418207)
(0.23500,0.4286706820)
(0.24000,0.4240537350)
(0.24500,0.4194893515)
(0.25000,0.4149759852)
(0.25500,0.4105121668)
(0.26000,0.4060964983)
(0.26500,0.4017276482)
(0.27000,0.3974043475)
(0.27500,0.3931253853)
(0.28000,0.3888896053)
(0.28500,0.3846959024)
(0.29000,0.3805432197)
(0.29500,0.3764305453)
(0.30000,0.3723569100)
(0.30500,0.3683213845)
(0.31000,0.3643230775)
(0.31500,0.3603611330)
(0.32000,0.3564347291)
(0.32500,0.3525430756)
(0.33000,0.3486854124)
(0.33500,0.3448610080)
(0.34000,0.3410691581)
(0.34500,0.3373091841)
(0.35000,0.3335804317)
(0.35500,0.3298822699)
(0.36000,0.3262140896)
(0.36500,0.3225753030)
(0.37000,0.3189653420)
(0.37500,0.3153836578)
(0.38000,0.3118297196)
(0.38500,0.3083030141)
(0.39000,0.3048030444)
(0.39500,0.3013293297)
(0.40000,0.2978814039)
(0.40500,0.2944588159)
(0.41000,0.2910611281)
(0.41500,0.2876879162)
(0.42000,0.2843387688)
(0.42500,0.2810132864)
(0.43000,0.2777110814)
(0.43500,0.2744317775)
(0.44000,0.2711750090)
(0.44500,0.2679404205)
(0.45000,0.2647276667)
(0.45500,0.2615364117)
(0.46000,0.2583663291)
(0.46500,0.2552171010)
(0.47000,0.2520884181)
(0.47500,0.2489799795)
(0.48000,0.2458914919)
(0.48500,0.2428226700)
(0.49000,0.2397732354)
(0.49500,0.2367429172)
(0.50000,0.2337314511)
(0.50500,0.2307385794)
(0.51000,0.2277640509)
(0.51500,0.2248076205)
(0.52000,0.2218690490)
(0.52500,0.2189481030)
(0.53000,0.2160445547)
(0.53500,0.2131581816)
(0.54000,0.2102887665)
(0.54500,0.2074360972)
(0.55000,0.2045999663)
(0.55500,0.2017801712)
(0.56000,0.1989765140)
(0.56500,0.1961888010)
(0.57000,0.1934168430)
(0.57500,0.1906604548)
(0.58000,0.1879194555)
(0.58500,0.1851936679)
(0.59000,0.1824829186)
(0.59500,0.1797870381)
(0.60000,0.1771058603)
(0.60500,0.1744392228)
(0.61000,0.1717869662)
(0.61500,0.1691489349)
(0.62000,0.1665249761)
(0.62500,0.1639149404)
(0.63000,0.1613186811)
(0.63500,0.1587360548)
(0.64000,0.1561669207)
(0.64500,0.1536111409)
(0.65000,0.1510685803)
(0.65500,0.1485391062)
(0.66000,0.1460225888)
(0.66500,0.1435189004)
(0.67000,0.1410279161)
(0.67500,0.1385495132)
(0.68000,0.1360835714)
(0.68500,0.1336299726)
(0.69000,0.1311886008)
(0.69500,0.1287593425)
(0.70000,0.1263420859)
(0.70500,0.1239367215)
(0.71000,0.1215431417)
(0.71500,0.1191612408)
(0.72000,0.1167909151)
(0.72500,0.1144320627)
(0.73000,0.1120845836)
(0.73500,0.1097483793)
(0.74000,0.1074233535)
(0.74500,0.1051094111)
(0.75000,0.1028064589)
(0.75500,0.1005144054)
(0.76000,0.0982331604)
(0.76500,0.0959626356)
(0.77000,0.0937027439)
(0.77500,0.0914533997)
(0.78000,0.0892145191)
(0.78500,0.0869860194)
(0.79000,0.0847678193)
(0.79500,0.0825598388)
(0.80000,0.0803619995)
(0.80500,0.0781742239)
(0.81000,0.0759964361)
(0.81500,0.0738285614)
(0.82000,0.0716705260)
(0.82500,0.0695222577)
(0.83000,0.0673836853)
(0.83500,0.0652547387)
(0.84000,0.0631353489)
(0.84500,0.0610254482)
(0.85000,0.0589249699)
(0.85500,0.0568338481)
(0.86000,0.0547520183)
(0.86500,0.0526794168)
(0.87000,0.0506159811)
(0.87500,0.0485616494)
(0.88000,0.0465163611)
(0.88500,0.0444800565)
(0.89000,0.0424526767)
(0.89500,0.0404341638)
(0.90000,0.0384244608)
(0.90500,0.0364235117)
(0.91000,0.0344312610)
(0.91500,0.0324476546)
(0.92000,0.0304726386)
(0.92500,0.0285061605)
(0.93000,0.0265481682)
(0.93500,0.0245986105)
(0.94000,0.0226574372)
(0.94500,0.0207245985)
(0.95000,0.0188000456)
(0.95500,0.0168837304)
(0.96000,0.0149756054)
(0.96500,0.0130756241)
(0.97000,0.0111837404)
(0.97500,0.0092999090)
(0.98000,0.0074240854)
(0.98500,0.0055562256)
(0.99000,0.0036962864)
(0.99500,0.0018442252)
(1.00000,0.0000000000)
};
\end{axis}
\end{tikzpicture}